\documentclass[11pt]{article}
\usepackage{geometry}
\usepackage{chngpage}
\usepackage{graphicx}
\usepackage{amsmath}
\usepackage{amsfonts}
\usepackage{amssymb}
\usepackage{amsthm}
\usepackage{latexsym}
\usepackage[english]{babel}
\usepackage[latin1]{inputenc}
\usepackage{multicol, blindtext}
\usepackage{times}
\usepackage[T1]{fontenc}
\usepackage{ulem} 
\usepackage{color}

\newcommand{\C}{\mathbb{C}}
\newcommand{\gln}{GL(n,\C)}
\newcommand{\gls}{GL(s+1,\C)}

\newtheorem{teo}{Theorem}
\newtheorem{cor}[teo]{Corollary}

\newtheorem{prop}[teo]{Proposition}
\newtheorem{defn}[teo]{Definition}
\newtheorem{Exemp}[teo]{Example}

\newtheorem{ob}[teo]{Remark}
\numberwithin{equation}{section}
\newcommand{\flip}{\operatorname{flip}}
\newcommand{\Flip}{\operatorname{Flip}}
\newcommand{\Oflip}{\operatorname{OFlip}}
\newcommand{\trace}{\operatorname{trace}}
\newcommand{\vrow}{\left[v_1 ; v_2 ; \cdots; v_n\right]}
\newcommand{\urow}{\left[u_1 ; u_2 ; \cdots; u_n\right]}

\begin{document}

\title{Polynomial Invariant Theory and Shape Enumerator of Self-Dual Codes in the NRT-Metric}
\author{Welington Santos \thanks{W. Santos is with the Programa de Pós-Graduação em Matemática, Universidade Federal do Paraná, Caixa Postal 19081, 81531-990, Curitiba-PR Brazil. His study was financed in part by the Coordenação de Aperfeiçoamento de Pessoal de Nível Superior- Brasil (CAPES)-Finance Code 001. Email: wsantos.math@gmail.com}~ Marcelo Muniz Silva Alves \thanks{Marcelo Muniz S. Alves is with the Departamento de Matemática, Universidade Federal do Paraná, Caixa Postal 19081, 81531-990, Curitiba-PR Brazil. This work was partially supported by CNPq-Conselho Nacional de Desenvolvimento Cient\'ifico e Tecnol\'ogico, research project number 306583/2016-0. Email: marcelomsa@ufpr.br}} 
\date{}

\maketitle
\begin{abstract}
\noindent In this paper we consider self-dual NRT-codes, that is, self-dual codes in the metric space endowed with the Niederreiter-Rosenbloom-Tsfasman (NRT-metric). 
We use polynomial invariant theory to describe the shape enumerator \cite{BargPark} of a binary self-dual, doubly even self-dual, and doubly-doubly even self dual NRT-code $C\subseteq M_{n,2}(\mathbb{F}_{2})$. 
Motivated by these results we describe the number of invariant polinomials that we must find to describe the shape enumerator of a self-dual NRT-code of $M_{n,s}(\mathbb{F}_{2})$. 
We define the ordered flip of a matrix $A\in M_{k,ns}(\mathbb{F}_{q})$ and present some constructions of self-dual NRT-codes over $\mathbb{F}_{q}$. We further give an application of ordered flip to the classification of bidimensional self-dual NRT-codes. 
\end{abstract}

\section{Introduction}

In classical coding theory, an $[n,k]$-linear code $C$ is a subspace of the metric space $\mathbb{F}_{q}^{s}$, endowed with the Hamming metric, which is  defined by $d_{H}(u,v):=|\lbrace i\ \text{such that}\ u_{i}\neq v_{i}\rbrace|$ for $u,v\in C\subseteq\mathbb{F}_{q}^{s}$. An important tool in this setting is the dual code $C^{\perp}$ of a linear code $C\subseteq\mathbb{F}_{q}^{s}$ with respect to usual (euclidean) inner product on $\mathbb{F}_{q}^{s}$.

One of the most important theorems in coding theory is MacWilliams' Theorem (1962), which is
known as the ``MacWilliams Identities'', which relates the weight enumerator of a linear code and its dual code. A remarkable theorem that is due to Gleason \cite{Glesonorig} shows that the weight enumerator of a binary doubly-even self-dual code is a polynomial in others two polynomials, to name, the weight enumerator of the Hamming code of length 8 and the Golay code of length 24.
During the past few years Sloane et al. presented some generalizations of Gleason's Theorem to other families of codes \cite{Rains,slone,NRS,SloaneNebe}. A technique that can be used to derive those generalizations is the polynomial invariant theory, and in special the well-known Molien's series of a finite group of matrices \cite{Tmolien}.

Coding theory has also been developed with respect to alternative metrics. 
One of the most studied of those metrics is the NRT-metric, which was introduced to study array codes that are subspaces of the linear space of all $n\times s$ matrices $M_{n,s}(\mathbb{F}_{q})$ with entries from a finite field. The  NRT-metric space was introduced by Rosenbloom and Tsfasman in \cite{NRT} by considering a generalization of Reed-Solomon codes on the space of all $n\times s$ matrices $M_{n,s}(\mathbb{F}_{q})$; in this same paper the authors pointed out that this metric models transmission over a set of parallel channels subject to fading. Since then several coding-theoretic questions with respect to this metric have been investigated, 
such as MacWilliams Identities \cite{DS} and MDS codes \cite{DS,4,Skriganov2001uniform-distributions}.  
Independently, Niederreiter \cite{H.Niederreiter} worked with a maximization problem in finite vector spaces which turned out to be equivalent to coding theory problems in NRT spaces, as was shown by Brualdi, Graves and Lawrence \cite{brualdi}.

Dougherty and Skriganov showed in \cite{DS} that the weight enumerators of mutually NRT-dual codes may not be related by any sort of MacWilliams Identities, in general, since there are examples of nonequivalent codes which have the same weight enumerators but whose duals have distinct weight enumerators. However, in the same paper they considered orbits of linear groups preserving the NRT-weight and showed that the weight enumerator associated with such orbits (called $H$-enumerator) satisfies a MacWilliams type identity for mutually NRT-dual codes.

Recently, Barg et al. \cite{BargAndPurka,Barg,BargAndFirer,BargPark} introduced the definition of shapes of codewords and a shape enumerator for NRT-codes. The shape enumerator coincides with the $H$-enumerator. Park and Barg in \cite{BargPark} obtained the same MacWilliams Identity of \cite{DS} with a new proof based on a multivariate Tutte polynomial of a NRT-code and centered in the concept of shape of a codeword.

In this contribution, we consider binary self-dual codes of $M_{n,s}(\mathbb{F}_{2})$, and we use polynomial invariant theory and the MacWilliams Identity of \cite{DS, BargPark} to describe the shape enumerator of those codes. In particular, for binary self-dual codes of $M_{n,2}(\mathbb{F}_{2})$ we completely describe their shape enumerator; the same is done for binary doubly-even self-dual codes, and for binary doubly-doubly-even self dual codes of $M_{n,2}(\mathbb{F}_{2})$. We describe the number of invariant polinomials that we must find to describe the shape enumerator of a self-dual NRT-code of $M_{n,2}(\mathbb{F}_{2})$. We define the concept of ordered flip of a matrix $A\in M_{k,ns}(\mathbb{F}_{q})$ and present some constructions of self-dual codes of $M_{n,s}(\mathbb{F}_{q})$ extending results for $M_{1,s}(\mathbb{F}_{q})$ of \cite{VMRS}. Finally, we present an application of the ordered flip to the classification of bidimensional self-dual codes.

This work is strutured as follows. In Sect. 2, we recall NRT-codes, the shape enumerator and the MacWilliams Identity of \cite{DS, BargPark}. In Sect. 3, we recall polynomial invariant theory and Molien's Theorem \cite{slone b, Hochater, ST2}. In Sect. 4, we describe the shape enumerator of binary self-dual, doubly even self-dual, and doubly-doubly even self dual NRT-codes of $M_{n,2}(\mathbb{F}_{2})$. In Sect. 5, we describe the Molien's series of the invariant group of binary self-dual NRT-codes of $M_{n,s}(\mathbb{F}_{q})$. In Sect. 6, we give some constructions of self-dual NRT-codes of $M_{n,s}(\mathbb{F}_{q})$ and an application of the flip concept.

\section{Codes and Niederreiter-Rosenbloom-Tsfasman Metric}
Let  $M_{n,s}(\mathbb{F}_{q})$ be the $\mathbb{F}_{q}$-vector space of $n\times s$ matrices with entries in $\mathbb{F}_{q}$. Given an $n \times s$ matrix $\mathbf{v}$, 
\[\mathbf{v}=\left[\begin{array}{cccc}
v_{11}&v_{12}&\cdots&v_{1s}\\
v_{21}&v_{22}&\cdots&v_{2s}\\
\vdots&\vdots&\ddots&\vdots\\
v_{n1}&v_{n2}&\cdots&v_{ns}
\end{array}\right] , 
\]
its $i$-th row will be denoted by $v_i$  and we will write 
$\mathbf{v}=\vrow$. Via this notation the Niederreiter-Rosenbloom-Tsfasman weight or, for short, the NRT-weigth of $\mathbf{v}$
is defined by the following formula: 
$$\rho(\mathbf{v}):=\sum_{i=1}^{n}\rho(v_{i})$$
where $\rho(v_{i}):=\max\left\lbrace   0 \leq  j \leq s   \  ;  v_{ij}\neq 0\right\rbrace$.

The canonical metric associated to the NRT- weight, $d_{\rho} (\mathbf{u}, \mathbf{v}) = \rho (\mathbf{u} -  \mathbf{v})$, is called the NRT-metric.

Now let $\mathbf{v}=\vrow \ \text{and}\ \mathbf{u}=\urow$ be two elements of $M_{n,s}(\mathbb{F}_{q})$. We define the inner product $\langle,\rangle_{N}$ on the space $M_{n,s}(\mathbb{F}_{q})$ endowed with the NRT-metric by
$$\langle\mathbf{v},\mathbf{u}\rangle_{N} =\langle\mathbf{u},\mathbf{v}\rangle_{N}:=\sum_{i=1}^{n}\langle v_{i},u_{i}\rangle_{N}$$
\noindent with 
$$\langle v_i, u_{i}\rangle_{N} =\langle u_{i}, v_{i}\rangle_{N}:=v_{i1}u_{is}+v_{i2}u_{is-1}+\ldots+v_{s-1 2}u_{i2}+v_{is}u_{i1}=\sum_{j=1}^{s} v_{ij}u_{i s+1-j}.$$
%
%

The space $\left(M_{ns}(\mathbb{F}_{q}),\rho\right)$ is called NRT-space and an $[ns,k]$-linear NRT-code is a linear subspace $C\subseteq\left(M_{ns}(\mathbb{F}_{q}),\rho\right)$ of dimension $k$.

\begin{defn}
The dual code of an $k$-dimensional linear NRT-code $C\subseteq M_{n,s}(\mathbb{F}_{q})$ is defined to be the $k^{\perp}$-dimensional linear NRT-code $C^{\perp}\subseteq M_{n,s}(\mathbb{F}_{q})$ given by
$$C^{\perp}:=\lbrace\mathbf{u}\in M_{n,s}(\mathbb{F}_{q}):\langle\mathbf{u},\mathbf{v}\rangle_{N}=0\ \text{for all}\ \mathbf{v}\in C\rbrace .$$
\end{defn}

\noindent A NRT-code $C$ is said to be a self-orthogonal NRT-code if $C\subseteq C^{\perp}$ and self dual NRT-code if $C=C^{\perp}$. Moreover, the dimensions of the codes $C$ and $C^{\perp}$ are related by the following equation: if $k = \dim(C)$ and $k^\perp = \dim(C^\perp)$ then 
\begin{equation}
k + k^\perp = ns.
\end{equation}

\subsection{The Geometry of NRT-metric}

Two NRT linear codes $C$ and $C'$ in $M_{n,s}(\mathbb{F}_{q})$ are equivalent if there is a linear isometry $\phi$ of $M_{n,s}(\mathbb{F}_{q})$
such that $\phi(C) = C'$. 
The group $GL\left(M_{n,s}(\mathbb{F}_{q})\right)$ of linear isometries of an NRT space $M_{n,s}(\mathbb{F}_{q})$ was described in  \cite{Lee} (also independently in \cite{DS}) and is isomorphic to the semidirect product   of $\left(T_{s}\right)^{n}$ and $S_{n}$,  where $T_{s}$ is the group of all upper triangular matrices of $M_{s,s}(\mathbb{F}_{q})$ with non-zero diagonal elements, $\left(T_{s}\right)^{n}$ denotes the direct product of $n$ copies of $T_{s}$, and $S_{n}$ is the symmetric group of order $n$.
An element of $S_n$ acts on  $\mathbf{v} = \vrow$  by permuting rows, and an element $(M_1,M_2, \ldots, M_n)$ of $T_s^n$ sends $\mathbf{v} = \vrow$ to $ [v_1 M_1^T; v_2 M_n^T; \cdots; v_nM_n^T]$. 

It is clear that if $\mathbf{v}=\vrow$ and $ \mathbf{u}=\urow$ 
are two elements of $M_{n,s}(\mathbb{F}_{q})$ that lie in the same $GL\left(M_{n,s}(\mathbb{F}_{q})\right)$-orbit then these matrices have the same NRT-weight. The converse will not hold in general, since $GL(M_{n,s}(\mathbb{F}_{q}))$ is not transitive on spheres.
In order to parametrize the $GL(M_{n,s}(\mathbb{F}_{q}))$-orbits, which were already studied in \cite{DS},
Barg and Purkayastha \cite{BargAndPurka} define a new parameter of the matrix $\mathbf{v}=\vrow \in M_{n,s}(\mathbb{F}_{q})$ which is called the \text{shape} of $\mathbf{v}$.

\begin{defn}
Let $\mathbf{v}\in M_{n,s}(\mathbb{F}_{q})$ be a matrix written as $\mathbf{v}=\vrow$, where $v_{i}=(v_{i1},\ldots,v_{is})$ for $i=1,\ldots,n$. The shape of $\mathbf{v}$ with respect to the NRT-weight is an $s$-vector $e=(e_{1},\ldots,e_{s})$, where
\[e_{j}=|\lbrace i\ \text{such that}\ 1\leqslant i\leqslant n\ \text{and}\ \rho(v_{i})=j\rbrace|. \]
Also define $e_{0}:=n-|e|$, where $|e|:=\sum_{j=1}^{s}e_{j}$.
\end{defn}
The action of $GL(M_{n,s}(\mathbb{F}_{q}))$ on matrices of a fixed shape is transitive, and then the shape is an invariant for this action. 
The NRT- weight can be also defined in terms of shapes: if $e=(e_{1},\ldots,e_{s})$ is the shape of the matrix $\mathbf{v}=\vrow \in M_{n,s}(\mathbb{F}_{q})$ then

$$\rho(\mathbf{v}):=\sum_{j=1}^{s}je_{j}.$$

\subsection{The shape enumerator and a MacWilliams Identity}
The NRT weight is a special case of \emph{poset weight} as introduced by Brualdi, Graves and Lawrence in
 \cite{brualdi}. 
 There is a notion of dual code for every poset code but only in rare cases does the weight enumerator of the code determine the enumerator of its dual; precisely, this is the case if and only if the poset is hierarchical 
 \cite{Kim-Oh}, and an NRT weight is associated to a hierarchical poset only when $n=1$ or $s=1$ (which corresponds to the Hamming weight). Nevertheless, analogues of those identities do hold when one considers
other kinds of enumerator polynomials.
In \cite{DS}, Dougherty and Skriganov defined a generalized weight enumerator for a NRT-linear code, the $H$-enumerator, which counts the number of codewords in each  
$GL(M_{n,s}(\mathbb{F}_{q}))$-orbit.
In the same paper it was shown that the $H$-enumerator of mutually dual codes satisfies a MacWilliams-type identity, which we will state next in the version presented in \cite{BargPark}, using the concept of shape  vector. 

A shape vector $e$ defines a partition of $ n$ into a sum of $s+1$ parts. We will denote by $\Delta_{s,n}:=\lbrace e\in\mathbb{N}^{s}:\ e_0 + e_{1}+e_{2}+\ldots+e_{s}\leqslant n\rbrace$ such partitions (recall that $e_0 = n - |e|$). 
In the language of shapes, the description of $GL(M_{n,s}(\mathbb{F}_{q}))$-orbits is as follows

\begin{prop}\textit{(\cite{DS}, Proposition 2.2 (ii))} 
The $GL(M_{n,s}(\mathbb{F}_{q}))$-orbits of a nonzero matrix $\mathbf{u} \in M_{n,s}(\mathbb{F}_{q})$ is the set of all vectors $\mathbf{v}\in M_{n,s}(\mathbb{F}_{q})$ which have the same shape as $\mathbf{u}$.
\end{prop}

\begin{defn}
Let $C\subseteq M_{n,s}(\mathbb{F}_{q})$ be a NRT-linear code. The shape enumerator of $C$ is the polynomial of $\C[z_0,z_1, \ldots, z_s]$ defined by
\begin{equation}\label{shapeenum}
H_{C}(z_0,z_1, \ldots, z_s)=\sum_{e\in\Delta_{s,n}}\mathcal{A}_{e}z_{0}^{e_0}z_{1}^{e_2}\ldots z_{s}^{e_s},
\end{equation}
where $\mathcal{A}_{e}=|\lbrace \mathbf{v} \in C\ \text{such that}\ \operatorname{shape}(\mathbf{v})=e\rbrace|$.
\color{black}
\end{defn}

The shape enumerator of a NRT-code $C$ is a homogeneous polynomial $H_{C}(Z)$ with $s+1$ variables which coincides with the $H$-enumerator introduced in \cite{DS}. In order to state the MacWilliams identity shown in $\cite{DS}$ we need to introduce new notation. 

The group $\gls$ of invertible $(s+1) \times (s+1)$ matrices acts on the ring 
$\C[z_0,z_1, \ldots, z_s]$ by 
\begin{equation}\label{action}
A \cdot f(z_0, z_1, \ldots, z_s) = f \left( \sum_{j=0}^{s} a_{0,j}z_j, \ldots, \sum_{j=0}^{s} a_{s,j}z_j\right).
\end{equation}
We may describe the action in a more concise manner. 
Consider the ``vector of variables'' $Z = (z_0,z_1, \ldots, z_s)^t$. The previous equation may be rewritten as 
\[
A \cdot f (Z) = f (A Z).
\]
Using this notation, the next result presents the MacWilliams identity for the shape enumerator. 
\color{black}

\begin{teo}\label{MacWilliamsH}\textit{\cite{BargPark}} The shape enumerator of mutually dual NRT-linear codes $C$ and $C^{\perp}\subseteq M_{n,s}(\mathbb{F}_{q})$ are related by
$$H_{C^{\perp}}(Z)=\frac{1}{|C|}H_{C}(\Theta_{s}Z),$$
where $\Theta_{s}= (\theta_{lk} ) \in M_{s+1,s+1} (\mathbb{F}_q)$ , $0\leqslant l, k\leqslant s$, has the following entries 
\[\theta_{lk}=\left\lbrace\begin{array}{ccl}
1& \textrm{ if } & l=0,\\
q^{l-1}(q-1)& \textrm{ if }& 0<l\leqslant s-k,\\
-q^{l-1}& \textrm{ if }& l+k=s+1,\\
0&\textrm{ if }& l+k>s+1.
\end{array}\right..\]
\end{teo}

\section{Polynomial Invariant Theory}

\subsection{Invariant Homogeneous Polynomial Basis} 

Let $G$ be a finite subgroup of $GL(n,\C)$ and consider its action on the polynomial ring 
$\C[x_1, \ldots, x_n]$ defined as in \eqref{action}.

We recall that $f(x_{1},\ldots ,x_{n})$ is called an invariant of $G$, or a $G$-invariant, if for every $A\in G$ 
\[
A\cdot f(x_{1},\ldots ,x_{n})=f(x_{1},\ldots , x_{n}).
\]

%
%
%
%

Clearly if $f$, $g$ are invariants  of $G$ so are $f+g$ and $fg$, it means that the set of invariants form a ring which will be denoted by $\mathcal{J}(G)$. 

It is well-known that every invariant polynomial can be written as a sum of  homogeneous invariant polynomials.
So given a finite group $G< M_{n,n}(\mathbb{C})$ it is enought to characterize all homogeneous polynomials that are invariant over $G$ to describe the invariant ring $\mathcal{J}(G)$ of $G$.

However, a convenient description of $\mathcal{J}(G)$ is a set of homogeneous invariants $B=\lbrace f_{1},\ldots, f_{l}\rbrace$ such that every polynomial in $\mathcal{J}(G)$ is a polynomial in $f_{1},\ldots, f_{l}$. Then in this case $B=\lbrace f_{1},\ldots ,f_{l}\rbrace$ is called a \textit{polynomial basis} for $\mathcal{J}(G)$. If $l>n$ there will be equations, which are called syzygies, relating $f_{1},\ldots, f_{l}$.

A finite polynomial basis of $\mathcal{J}(G)$ always exists when $G$ is finite, this was shown by  E. Noether's Theorem \cite{slone,ST,ST2}. Find invariant polynomials is fairly easy using the following theorem.
\begin{teo}\label{RICHARd}\textit{(Reynolds operator)} Let $f\in\mathbb{C}[x_{1},\ldots,x_{n}]$ be a polynomial and $G<\gln$ a finite group. Then 
	$$h(x_{1},\ldots ,x_{n})=\sum_{A\in G}A\cdot f(x_{1},\ldots ,x_{n})$$
	is an invariant of $G$, that is, $h\in\mathcal{J}(G)$.
\end{teo}
 

\begin{defn} Given a finite group $G<\gln$, a  good polynomial basis for $\mathcal{J}(G)$ consists of 
	homogeneous invariants $f_{1},\ldots ,f_{l}$ $(l\geq n)$ such that $f_{1},\ldots f_{n}$ are algebraically independent and
	\begin{equation*}\label{base1}
	\mathcal{J}(G)=\mathbb{C}[f_{1},\ldots,f_{n}]\quad \textrm{ if } \quad l=n,
	\end{equation*}
	or, if $l>n$,
	\begin{equation*}\label{base2}
	\mathcal{J}(G)=\mathbb{C}[f_{1},\ldots,f_{n}]\oplus
	f_{n+1}\mathbb{C}[f_{1},\ldots,f_{n}]\oplus \cdots\oplus f_{l}\mathbb{C}[f_{1},\ldots,f_{n}].
	\end{equation*}	
\end{defn}

\noindent This means that when $l=n$ any $G$-invariant can be written as a polynomial in $f_{1},\ldots,f_{n}$ and, if $l > n$,  as such a polynomial plus $f_{n+1}$ times another such polynomial and so on. The polynomials $f_{1},\ldots,f_{n}$ are called primary invariants and $f_{n+1},\ldots,f_{l}$ are secondary invariants.


\begin{teo} \textit{(Hochster and Eagon \cite{Hochater})} A good polynomial basis exists for
	the invariants of any finite group of complex $n \times n$ matrices. 
\end{teo}

There are some criteria in literature that help to find out when a set of polynomials is algebraically independent. In this work we use a useful tool that is know as the Jacobian criterion.

\begin{teo}\textit{(Jacobian criterion \cite{Lefs,Richard})}\label{Jacobian} $f_{1},\ldots,f_{l}\in\mathbb{C}[x_{1},\ldots,x_{n}]$ are algebraically independent if only if the Jacobian matrix is full rank. 
In particular, if $m=n$, $f_{1},\ldots,f_{l}$ are algebraically dependent if only if the Jacobian determinant $\det(J(f_{1},\ldots,f_{l}))\neq 0$.
\end{teo}

\subsection{Molien's Theorem}
A fundamental problem  is to know, or at least estimate, 
how many algebraically independent invariants are required to form a polynomial basis of $\mathcal{J}(G)$.
The next two theorems tell us how many linearly independent homogeneous invariants exist for each degree $t$.

\begin{teo}\textit{(Molien \cite{Tmolien})}\label{invgrau1} Given a finite group $G<
\gln $ the number of linearly independent invariants over $G$ of the first degree is
$$\frac{1}{|G|}\sum_{A\in G}\trace(A).$$
\end{teo}

\begin{teo}\textit{(Molien \cite{Tmolien})} Let $G<\gln$ be a finite group. Then, the number of linearly independent invariants of $G$ of degree $t$ is the coefficient of $\lambda^{t}$ in the expansion of
$$\Phi_{G}(\lambda)=\frac{1}{|G|}\sum_{A\in G}\frac{1}{\det(I-\lambda A)}.$$
$\Phi_{G}(\lambda)$ is called the Molien series of $G$.
\end{teo}
The Molien series of a finite group $G<\gln$ can be written as 
$$\Phi_{G}(\lambda)=\frac{1}{|G|}\sum_{A\in G}\frac{\det(A)}{\det(A-\lambda I)}.$$

Another important fact about the Molien series of a finite group $G<\gln$ is that it can be written down by inspection from the degrees of a good polynomial basis, that is, if $B=\lbrace f_{1},\ldots ,f_{l}\rbrace$ is a good polynomial basis of $\mathcal{J}(G)$ such that $d_{1}=\deg f_{1},\ldots ,d_{l}=\deg f_{l}$ then the Molien series of $G$ can be written as
\begin{equation}\label{moli1}
\Phi_{G}(\lambda)=\frac{1}{\prod_{1}^{n}(1-\lambda^{d_{i}})},\quad if\quad l=n,
\end{equation}
or
\begin{equation}\label{moli2}
\Phi_{G}(\lambda)=\frac{1+\sum_{j=l+1}^{n}\lambda^{d_{j}}}{\prod_{i=1}^{n}(1-\lambda^{d_{i}})}\quad if\quad l>n.
\end{equation}

On the other hand, the converse is not true. It is not always true when the Molien series has been put into the forms ($\ref{moli1}$) or ($\ref{moli2}$) then a good polynomial basis for $\mathcal{J}(G)$ can be found with degrees matching the powers of $\lambda$ in $\Phi(\lambda)$. This was shown by an example, due to Stanley \cite[Ex 3.8]{ST2}.

\section{Invariant Theory and the  shape enumerator}
In this section, we investigate the shape enumerator $H_{C}$ of binary self-dual NRT-code $C\subseteq M_{n,s}(\mathbb{F}_{2})$. 
\subsection{Invariant Ring for Self-dual NRT-Codes of $M_{n,2}(\mathbb{F}_{2})$}
We know from Theorem \ref{MacWilliamsH} that the shape enumerators of mutually dual linear codes $C$ and $C^{\perp}\subseteq M_{n,2}(\mathbb{F}_{q})$ are related by the equation
\begin{equation} \label{equation.MW.2}
H_{C^{\perp}}(Z)=\frac{1}{2^{d}}H_{C}(\Theta_{2}Z), 
\end{equation}
where 
\[
\Theta_{2}=\left[\begin{array}{crr}
1&1&2\\
1&1&-2\\
1&-1&0
\end{array}\right]. 
\]

If we assume that $C$ is a linear self-dual NRT-code of dimension $k$, that is, $C=C^{\perp}$ and $k=k^{\perp}$, then we conclude that $k=n$, so Equation \eqref{equation.MW.2} can be rewritten as 
$$H_{C}(Z)=\frac{1}{2^{n}}H_{C}(\Theta_{2}Z).$$
Since by definition $H_{C}(Z)$ is a homogeneous polynomial of degree $n$, the last expression for $H_{C}(Z)$ may be rewritten as 
$$H_{C}(Z)=H_{C}\left(\frac{\Theta_{2}}{2}Z\right),$$
which means that the polynomial $H_{C}$ is invariant by $T=\frac{\Theta_{2}}{2}$. Moreover, as $T^{2}=\left(\frac{\Theta_{2}}{2}\right)^{2}=I_{3}$, the shape enumerator $H_{C}$ is an element of  $\mathcal{J}(G_{1})$ where $G_1$ is the finite group of order 2 generated by $T$, $G_{1}=\langle T\rangle=\lbrace I,T\rbrace$. We will try construct a good polynomial basis for $\mathcal{J}(G_{1})$.

From Molien's Theorem the number of algebraically independent invariants of degree $t$ over the group $G_{1}$ is equal to the coefficient of $\lambda^{t}$ in

$$\Phi_{G_{1}}(\lambda)=\frac{1}{2}\sum_{A\in G_{1}}\frac{det(A)}{det(A-\lambda I_{3})}.$$
Let's calculate $\Phi_{G_{1}}(\lambda)$. First note that 
\begin{equation}\label{matrixT}
T:=\frac{\Theta_{2}}{2}=\left[\begin{array}{crr}
\frac{1}{2}&\frac{1}{2}&1\\
\frac{1}{2}&\frac{1}{2}&-1\\
\frac{1}{2}&-\frac{1}{2}&0
\end{array}\right].
\end{equation}
So the Molien's series of $G_{1}$ is given by

\begin{equation*}
\Phi_{G_{1}}(\lambda)=\frac{1}{2}\sum_{A\in G_{1}}\frac{det(A)}{det(A-\lambda I_{3})}=\frac{1}{(1-\lambda)^{2}(1-\lambda^{2})}.
\end{equation*}

It means that to find a good polynomial basis for $G_{1}$ we should look for one invariant of degree two and two invariants of degree one.

 Since the shape enumerator of any NRT-code $C\subseteq M_{n,2}(\mathbb{F}_{2})$ has degree $n$, we will start looking for shape enumerators of self-dual NRT-codes in $M_{1,2}(\mathbb{F}_{2})=\lbrace (0,0),(1,0),(0,1),(1,1)\rbrace$. In $M_{1,2}(\mathbb{F}_{2})$ there exist just five linear codes, namely: The trivial codes $C_{0}:=\lbrace (0,0)\rbrace$, $C_{4}:= M_{1,2}(\mathbb{F}_{2})$ and the non-trivial codes:
\[ C_{1,1}:=\lbrace (0,0),(0,1)\rbrace;\]
\[ C_{1,2}:=\lbrace (0,0),(1,0)\rbrace;\]
\[ C_{1,3}:=\lbrace (0,0),(1,1)\rbrace;\]
It is clear that all except the trivial codes are self-dual codes, and their shape enumerators are
\[H_{C_{1,1}}\left(z_{0},z_{1},z_{2}\right)=z_{0}+z_{2};\]
\[H_{C_{1,2}}\left(z_{0},z_{1},z_{2}\right)=z_{0}+z_{1};\]
\[H_{C_{1,3}}\left(z_{0},z_{1},z_{2}\right)=z_{0}+z_{2}.\]
Note that $H_{C_{1,2}}\left(z_{0},z_{1},z_{2}\right)=H_{C_{1,3}}\left(z_{0},z_{1},z_{2}\right)$ and this was already expected once that there exists a linear isometry between $C_{1,2}$ and $C_{1,3}$. We choose $\phi_{1}(z_{0},z_{1},z_{2})=z_{0}+z_{2}$ and $\phi_{2}(z_{0},z_{1},z_{2})=z_{0}+z_{1}$; it is obvious that $\phi_{1},\phi_{2}$ are algebraically independent and invariant under $G_{1}$.

Now to find an invariant polynomial of degree two we will consider a self-dual NRT-code in $M_{2,2}(\mathbb{F}_{2})$ and compute its shape enumerator. Let
\[ C_{2,1}:=\left\lbrace \left[\begin{array}{cc}
0&0\\
0&0
\end{array}\right],\left[\begin{array}{cc}
1&0\\
1&0
\end{array}\right],\left[\begin{array}{cc}
0&1\\
0&1
\end{array}\right], \left[\begin{array}{cc}
1&1\\
1&1
\end{array}\right]\right\rbrace.\]

$C_{2,1}$ is a self-dual code and its shape enumerator is $H_{C_{2,1}}\left(z_{0},z_{1},z_{2}\right)=z_{0}^{2}+z_{1}^{2}+2z_{2}^{2}$. Define $\phi_{3}(z_{0},z_{1},z_{3})=z_{0}^{2}+z_{1}^{2}+2z_{2}^{2}$, so $\phi_{3}(z_{0},z_{1},z_{3})$ is invariant under $G_{1}$. Moreover, applying Theorem $\ref{Jacobian}$, we can check that the set $\lbrace\phi_{1},\phi_{2},\phi_{3}\rbrace$ is algebraically independent. In short, we just proved the following theorem.

\begin{teo}\label{Invmn2} Let $C\subseteq M_{n,2}(\mathbb{F}_{2})$ be a self-dual NRT-code. Then, the shape enumerator of $C$ is an invariant polynomial under the action of $G_{1}=\langle T\rangle$ where $T$ is the matrix given in (\ref{matrixT}). Moreover, the invariant ring of the group $G_{1}$ is $\mathbb{C}[\phi_{1},\phi_{2},\phi_{3}]$ where $\phi_{1}(z_{0},z_{1},z_{2})=z_{0}+z_{2}$, $\phi_{2}(z_{0},z_{1},z_{2})=z_{0}+z_{1}$, and $\phi_{3}(z_{0},z_{1},z_{2})=z_{0}^{2}+z_{1}^{2}+2z_{2}^{2}$.
\end{teo}
\noindent In words, Theorem \ref{Invmn2} means that the shape enumerator of any linear self-dual NRT-code $C\subseteq Mat_{n,2}(\mathbb{F}_{2})$ is a polynomial in $\phi_{1},\phi_{2}$ and $\phi_{3}$.

\subsection{Invariant Ring for Doubly-Even Self-Dual NRT-Codes of $M_{n,2}(\mathbb{F}_{2})$}
Let now $C\subseteq M_{n,2}(\mathbb{F}_{2})$ be a doubly-even self-dual NRT-code, i.e., a self-dual NRT-code $C$ whose every codeword has even weight. From the definition of shape enumerator of $C$ it follows that $H_{C}(z_{0},z_{1},z_{2})\in\mathbb{C}[z_{0},z_{1},z_{2}]$ is such that $z_{1}$ has always even degree. So in this case 
\[H_{C}(z_{0},z_{1},z_{2})=H_{C}(z_{0},-z_{1},z_{2}).\]
Thats is, the polynomial $H_{C}$ is invariant by $$A:=\left[\begin{array}{crc}
1&0&0\\
0&-1&0\\
0&0&1
\end{array}\right].$$

Since $C$ is a self-dual NRT-code, we already know that $H_{C}$ is invariant by the matrix $T$ given in (\ref{matrixT}). Finally, the shape enumerator of $C$ is invariant under the group $G_{2}:=\left\langle T,A\right\rangle$. It is possible to check that \begin{equation}\label{G2}
G_{2}=\lbrace I, A, T, AT, TA, TAT\rbrace\end{equation}  and the order of $G_{2}$ is $g=6$. The Molien's series of $G_{2}$ is given by

\begin{equation*}
\Phi(\lambda)=\frac{1}{6}\sum_{A\in G_{2}}\frac{det(A)}{det(A-\lambda I_{3})}=\frac{1}{(1-\lambda)(1-\lambda^{2})(1-\lambda^{3})}.
\end{equation*}

It suggests that we should search for one invariant of degree one, one invariant of degree two and one invariant of degree three in order to determine a polynomial basis of invariants.
 The code $C_{1,1}=\lbrace (0,0),(0,1)\rbrace$ is such that $C_{1,1}$ is a self-dual NRT-code and all its codewords has even weight. Moreover, $H_{C_{1,1}}\left(z_{0},z_{1},z_{2}\right)=z_{0}+z_{2}$ is a $G_2$-invariant. Now consider $C_{2,2}\subseteq M_{2,2}(\mathbb{F}_{2})$ given by

\[C_{2,2}=\left\lbrace\left[\begin{array}{cc}
0&0\\
0&0
\end{array}\right],\left[\begin{array}{cc}
1&0\\
1&0
\end{array}\right],\left[\begin{array}{cc}
0&1\\
0&1
\end{array}\right],\left[\begin{array}{cc}
1&1\\
1&1
\end{array}\right]\right\rbrace.\]
$C_{2,2}$ is a self-dual NRT-code and all its codewords have even weight. Furthermore, the shape enumerator of $C_{2,2}$ is the polynomial $H_{C_{2}}(z_{0},z_{1},z_{2})=z_{0}^{2}+z_{1}^{2}+2z_{2}^{2}$ which is, of course, an invariant polynomial of the group $G_{2}$. Putting $p_{1}$ and $p_{2}$ as $p_{1}(z_{0},z_{1},z_{2})=z_{0}+z_{2}$ and $p_{2}(z_{0},z_{1},z_{2})=z_{0}^{2}+z_{1}^{2}+2z_{2}^{2}$ we have an algebraically independent set $\left\lbrace p_{1},p_{2}\right\rbrace$. So, we just need do find another polynomial $p_{3}(z_{0},z_{1},z_{2})$ such that $p_{3}$ is algebraically independent of $p_{1}$ and $p_{2}$.

Define $C_{3,3}\subseteq M_{3,2}(\mathbb{F}_{2})$ by
$$\
C_{3,3}:=\left\{\left[\begin{array}{cc}
0&0\\
0&0\\
0&0
\end{array}\right], \left[\begin{array}{cc}
0&0\\
1&0\\
1&0
\end{array}\right],\left[\begin{array}{cc}
1&1\\
1&1\\
0&1
\end{array}\right],\left[\begin{array}{cc}
1&0\\
0&0\\
1&0
\end{array}\right],
\left[\begin{array}{cc}
1&1\\
0&1\\
1&1
\end{array}\right],\left[\begin{array}{cc}
0&1\\
1&1\\
1&1
\end{array}\right],\left[\begin{array}{cc}
1&0\\
1&0\\
0&0
\end{array}\right],\left[\begin{array}{cc}
0&1\\
0&1\\
0&1
\end{array}\right]\right\}$$
$C_{3,3}$ is a self-dual NRT-code and all its codewords has even weight. Moreover, the shape enumerator of $C_{3,3}$ is the polynomial $H_{C_{3,3}}(z_{0},z_{1},z_{2})=z_{0}^{3}+4z_{2}^{3}+3z_{1}^{2}z_{0}$ that is invariant under $G_{2}$. Defining $p_{3}(z_{0},z_{1},z_{2})=H_{C_{3,3}}(z_0,z_1,z_2)$ the $\lbrace p_{1}, p_{2}, p_{3}\rbrace$  is algebraically independent by Theorem $\ref{Jacobian}$.


Summing up, we have just proved the following.
\begin{teo} Let $C\subseteq M_{n,2}(\mathbb{F}_{2})$ be a self-dual NRT-code such that all its codewords has even weight. Then, the shape enumerator of $C$ is an invariant polynomial for the group $G_{2}$ given in (\ref{G2}). Moreover, the invariant ring of $G_{2}$ is $\mathbb{C}[p_{1},p_{2},p_{3}]$ where the polynomials $p_{1},p_{2}$ and $p_{3}$ are given by $p_{1}(z_{0},z_{1},z_{2})=z_{0}+z_{2}$, $p_{2}(z_{0},z_{1},z_{2})=z_{0}^{2}+z_{1}^{2}+2z_{2}^{2}$, and $p_{3}(z_{0},z_{1},z_{2})=z_{0}^{3}+4z_{2}^{3}+3z_{1}^{2}z_{0}$.
\end{teo}
\noindent In other words, if $C\subseteq M_{n,2}(\mathbb{F}_{2})$ is a self-dual NRT-code whose codewords has even weight then its shape enumerator is a polynomial in $p_{1}, p_{2}$ and $p_{3}$.

\subsection{Invariant Ring for Doubly-Doubly-Even Self-Dual NRT-Codes of $M_{n,2}(\mathbb{F}_{2})$}

Let $C\subseteq M_{n,2}(\mathbb{F}_{2})$ be a self-dual NRT-code whose every codeword has an even number of rows with weight one and an even number of rows with weight two. In this case, by definition of shape enumerator, $H_{C}(z_{0},z_{1},z_{2})$ is such that $z_{1}$ and $z_{2}$ are always of degree even, therefore, it is true that $H_{C}(z_{0},z_{1},z_{2})=H_{C}(z_{0},-z_{1},-z_{2})$.
This implies that $H_{C}(z_{0},z_{1},z_{2})$ the shape enumerator of $C$ is invariant under the action of the matrix
\[B:=\left[\begin{array}{crr}
1&0&0\\
0&-1&0\\
0&0&-1
\end{array}\right].\]
Since $C$ is a self-dual NRT-code it follows also that $H_{C}$ is invariant by the matrix $T$ defined in (\ref{matrixT}), and therefore the polynomial $H_C$ is invariant under the action of the group 
\begin{equation}\label{G3}
G_{3}:=\langle T,B\rangle.\end{equation}
We can check that $G_{3}=\lbrace I, B, T, BT, TB, TBT, BTB, (BT)^{2}, (TB)^2, (TB)^{3}, B(TB)^{2}, (TB)^{2}T\rbrace$ and 
\[\frac{1}{|G_{3}|}\sum_{A_{i}\in G_{3}} \trace(A_{i})=0.\]
Therefore, Theorem \ref{invgrau1} assure us that there are no $G_3$-invariants of degree one. We can write the Molien's series of $G_{3}$ as
\begin{eqnarray*}
	\Phi_{G_{3}}(\lambda)&=&\frac{1}{12}\left(\frac{12+12\lambda^{4}}{(1-\lambda^{2})^{2}(1-\lambda^{6})}\right)\\
	&=&\frac{1}{(1-\lambda^{2})^{2}(1-\lambda^{6})}+\frac{\lambda^{4}}{(1-\lambda^{2})^{2}(1-\lambda^{6})},
\end{eqnarray*}
which suggests that in order obtain a good polynomial basis for $\mathcal{J}(G_{3})$ we should search for three primary invariants $\phi_{1},\phi_{2}$ and $\phi_{3}$ of degree $2,2,6$ respectively and one secondary invariant $\phi_{4}$ of degree $4$. 

Unfortunately, there exists only one linear code $C\subseteq M_{2,2}(\mathbb{F}_{2})$ such that the required properties are satisfied, namely $C_{2,2}$ of the previous section, so we can take $\phi_{1}$ as $\phi_{1}(z_{0},z_{1},z_{2})=z_{0}^{2}+z_{1}^{2}+2z_{2}^{2}$. By averaging $z_{0}^{2}$ under the group $G_{3}$, using the Theorem $\ref{RICHARd}$, we obtain the invariant $\phi_{2}(z_{0},z_{1},z_{2})=5z_{0}^{2}-2z_{0}z_{1}+z_{1}^{2}+8z_{2}^{2}+8z_{2}z_{1}$. 

Now, we will work to find an invariant of degree six. Averaging $z_{1}z_{2}$ over the group $G_{3}$, using Theorem $\ref{RICHARd}$, we obtain an homogeneous invariant $\phi_{3}^{\star}$ of degree two, namely $\phi_{3}^{\star}(z_{0},z_{1},z_{2})=2z_{0}^{2}-2z_{1}^2+8z_{1}z_{2}$. The set $\lbrace
\phi_{1},\phi_{2},\phi_{3}^{\star}\rbrace$ is algebraically independent. Let $\phi_{3}\in\mathbb{C}[z_{0},z_{1},z_{2}]$ be the polynomial given by $\phi_{3}=(\phi_{3}^{\star})^{3}$, so $\deg\phi_{3}=6$ and $\lbrace\phi_{1},\phi_{2},\phi_{3}\rbrace$ is algebraically independent since $\lbrace
\phi_{1},\phi_{2},\phi_{3}^{\star}\rbrace$ is algebraically independent. We can use Magma 
Computer Algebra program \cite{Magma} to find the secondary invariant $\phi_{4}$.

\begin{teo}
Let $C\subseteq M_{n,2}(\mathbb{F}_{2})$ be a self-dual NRT-code such that all its codewords has an even number of rows with weight one, and an even number of rows with weight two. Then, the shape enumerator of $C$ is an invariant polynomial for the group $G_{3}$ given in (\ref{G3}). Moreover, the invariant ring of $G_{3}$ is $\mathbb{C}[\phi_{1},\phi_{2},\phi_{3}]\oplus\phi_{4}\mathbb{C}[\phi_{1},\phi_{2},\phi_{3}]$ where the polynomials $\phi_{1},\phi_{2}$ and $\phi_{3}$ are $\phi_{1}(z_{0},z_{1},z_{2})=z_{0}^{2}+z_{1}^{2}+2z_{2}^{2}$, $\phi_{2}(z_{0},z_{1},z_{2})=5z_{0}^{2}-2z_{0}z_{1}+z_{1}^{2}+8z_{2}^{2}+8z_{2}z_{1}$, $\phi_{3}(z_{0},z_{1},z_{2})=\left(2z_{0}^{2}-2z_{1}^{2}+8z_{1}z_{2}\right)^{2}$.	
\end{teo}
\section{The General case of Self-dual NRT-Codes of $M_{n,s}(\mathbb{F}_{2})$}

For the purpose of studying the shape enumerator of a self dual NRT-code $C\subseteq M_{n,s}(\mathbb{F}_{2})$ we will first look closely at the matrix $\Theta_{s}$ given in Theorem \ref{MacWilliamsH}. 

\begin{teo}\label{propT}The matrix $\Theta_{s}$ given by $\Theta_{s}=(\theta_{l,k})_{l,k=0,\ldots,s}$ where
	\[\theta_{l,k}:=\left\lbrace\begin{array}{rcc}
	1&\textrm{if}&k=0,\\
	2^{k-1}&\textrm{if}&0<k\leqslant s-l,\\
	-2^{k-1}&\textrm{if}&l+k=s+1,\\
	0&\textrm{if}&l+k>s+1,
	\end{array}\right. \]
satisfies the following properties:
\begin{itemize}
	\item[a)] $\Theta_{s}^{2}=2^{s}I_{s+1}$;
	\item[b)] $\trace(\Theta_{s})=\left\lbrace\begin{array}{lll}
	2^{\frac{s}{2}}&\textrm{if}& s\quad\textrm{is}\quad \textrm{even}\\
	0&\textrm{if}& s\quad \textrm{is}\quad \textrm{odd}
	\end{array}\right. ;$
	\item[c)] $\det(\Theta_{s}):=\left\lbrace\begin{array}{lll}
	(-1)^{\frac{s+1}{2}}2^{\frac{s(s+1)}{2}}&\textrm{if}& s\quad \textrm{is}\quad \textrm{odd}\\
	(-1)^{\frac{s}{2}}2^{\frac{s(s+1)}{2}}&\textrm{if}& s\quad \textrm{is}\quad \textrm{even}
	\end{array}\right. ;$
\end{itemize}
\end{teo}
\begin{proof}
Item a) is immediate and also is a fact already noticed in \cite{DS}. So we will prove items b) and c).
	
\begin{itemize}
	\item[b)]
	By the definition of $\Theta$ we have that $\displaystyle \trace(\Theta_{s})=\sum_{i=0}^{s}\theta_{ii}$  where
	\[\theta_{ii}:=\left\lbrace\begin{array}{rrl}
	1&\textrm{if}&i=0\\
	2^{i-1}&\textrm{if}&0<i\leqslant s-i\\
	-2^{i-1}&\textrm{if}&2i=s+1\\
	0&\textrm{if}&2i>s+1
	\end{array}\right. .\]
	So for an even $s$ the trace of $\Theta_{s}$ becomes
	\begin{eqnarray*}
		\trace(\Theta_{s})&=&\displaystyle\sum_{i=0}^{s}\theta_{ii}
		=1+\displaystyle\sum_{i=1}^{\frac{s}{2}}2^{i-1}+ 0 \\
		&=&2^{\frac{s}{2}}.
	\end{eqnarray*}
	On the other hand, for a odd $s$ the trace of $\Theta_{s}$ will be
	\begin{eqnarray*}
		\trace(\Theta_{s})&=&\displaystyle\sum_{i=0}^{s}\theta_{ii}
		=1+\displaystyle\sum_{i=1}^{\frac{s-1}{2}}\theta_{ii}+\theta_{\frac{s+1}{2},\frac{s+1}{2}}+\displaystyle\sum_{i=\frac{s+1}{2}+1}^{s}\theta_{ii}\\
		&=&1+(2^{\frac{s-1}{2}}-1)-2^{\frac{s-1}{2}}+0\\
		&=&0,
	\end{eqnarray*}
and we have just proved that
$$\trace(\Theta_{s})=\left\lbrace\begin{array}{lll}
2^{\frac{s}{2}}&\textrm{if}& s\quad\textrm{is}\quad \textrm{even,}\\
0&\textrm{if}& s\quad \textrm{is}\quad \textrm{odd.}
\end{array}\right. $$
\item[c)]
	Setting $\Theta_{0}=1$ and applying Laplace expansion along the last column of $\Theta_{s}$ we find that 
	\begin{equation}\label{laplace}	
	\det(\Theta_{s})=2^{s-1}(-1)^{s}\det(\Theta_{s-1})-2^{s-1}(-1)^{s+1}\det(\Theta_{s-1}) = 2^s (-1)^s \det(\Theta_{s-1}) .
	\end{equation}
	for every $s \geq 1$. 
	
	This equation yields the two equalities below:
	\begin{eqnarray}
		\det(\Theta_{2t+2}) &=& 2^{2t+2}\det(\Theta_{2t+1}), \label{2t+2.2t+1}\\
		\det(\Theta_{2t+1})&=&-2^{2t+1}\det(\Theta_{2t}) \label{2t+1.2t}
	\end{eqnarray}
	for every $t \geq 0$. 
	Using \eqref{2t+2.2t+1} and \eqref{2t+1.2t}, and taking $s=2t$  we obtain, by induction on $t$, that the equality
	\begin{equation} \label{det.even.s}
	\det(\Theta_{s})=(-1)^{\frac{s}{2}}2^{\frac{s(s+1)}{2}}
	\end{equation} 
	holds for any even number $s$. Now it follows from the previous equation and from \eqref{2t+1.2t} that 
	\begin{equation}\label{det.odd.s}
	\det(\Theta_{s})=(-1)^{\frac{s+1}{2}}2^{\frac{s(s+1)}{2}}
	\end{equation}
	for any odd number $s$, concluding the proof of $c)$. 
\end{itemize}
\end{proof}

\begin{teo}\label{caracTheta} Let $m_{\Theta_{s}}$ and 
and $P_{\Theta_{s}}$ be the minimal and characteristic polynomials respectively of the matrix 
$\Theta_{s}$ given in Theorem \ref{propT}. Then $m_{\Theta_{s}}$ and $P_{\Theta_{s}}$ are given by 
$m_{\Theta_{s}}(\lambda)=(\lambda-2^{\frac{s}{2}})(\lambda+2^{\frac{s}{2}})$ and
$$p_{\Theta_{s}}(\lambda)=\left\lbrace\begin{array}{lll}
			(\lambda-2^{\frac{s}{2}})^{\frac{s+2}{2}}(\lambda+2^{\frac{s}{2}})^{\frac{s}{2}}&\textrm{if}& s\quad \textrm{is}\quad \textrm{even}\\
		(\lambda-2^{\frac{s}{2}})^{\frac{s+1}{2}}(\lambda+2^{\frac{s}{2}})^{\frac{s+1}{2}}&\textrm{if}& s\quad \textrm{is}\quad \textrm{odd}
		\end{array}\right. $$
\end{teo}
\begin{proof}
 Item a) of Theorem \ref{propT} says that
 \[ 0 = \Theta_{s}^{2}-2^{s}I_{s+1}= (\Theta_{s}-2^{\frac{s}{2}}I_{s+1})(\Theta_{s}+2^{\frac{s}{2}}I_{s+1})
 \]
  and it follows that $m_{\Theta_{s}}(\lambda)=(\lambda-2^{\frac{s}{2}})(\lambda+2^{\frac{s}{2}})$ is the minimal polynomial of $\Theta_{s}$, since $m_{\Theta_{s}}$ is a monic polynomial of degree two such that $m_{\Theta_{s}}(\Theta_{s})=0$ and obviously no polynomial of degree one vanishes on $\Theta_{s}$. Therefore the characteristic polynomial of $\Theta_{s}$ decomposes as a product
\begin{equation}\label{carcT}
p_{\Theta_{s}}(\lambda)=(\lambda-2^{\frac{s}{2}})^{r_{1}}(\lambda+2^{\frac{s}{2}})^{r_{2}},
\end{equation}
where $r_{1}$ and $r_{2}$ are the multiplicities of the eigenvalues $\beta_{1}=2^{\frac{s}{2}}$ and $\beta_{2}=-2^{\frac{s}{2}}$ and, in particular,  $r_1 + r_2 = s+1$.
From \eqref{carcT} it follows that 
\begin{equation}
\trace(\Theta_{s}) = r_{1}(2^{\frac{s}{2}})+r_{2}(-2^{\frac{s}{2}}),
\end{equation}
and therefore for every $s \geq 1$ we have the system of equations 
\[\left\lbrace\begin{array}{ccc}
r_{1}+r_{2}&=&s+1\\
r_{1}-r_{2}&=&\dfrac{\trace(\Theta_{s})}{2^{\frac{s}{2}}} 
\end{array}\right. 
\]

If $s$ is an odd number then $\trace(\Theta_{s}) = 0$ by 
item $b)$ of Theorem \ref{propT} and, in this case, the above system has $r_{1}=r_{2}=\frac{s+1}{2}$ as unique solution. 
Therefore in the case of an odd $s$ the characteristic polynomial of $\Theta_{s}$ is $$p_{\Theta_{s}}(\lambda)=(\lambda-2^{\frac{s}{2}})^{\frac{s+1}{2}}(\lambda+2^{\frac{s}{2}})^{\frac{s+1}{2}}.$$

If  $s$ is even then $\trace(\Theta_{s}) = 2^{\frac{s}{2}}$ and the corresponding system has solution 
$r_{1}=\frac{s+2}{2}, r_{2}=\frac{s}{2}$. In this case the characteristic polynomial of $\Theta_{s}$ is $$p_{\Theta_{s}}(\lambda)=(\lambda-2^{\frac{s}{2}})^{\frac{s+2}{2}}(\lambda+2^{\frac{s}{2}})^{\frac{s}{2}}.$$
We have just proved that 
$p_{\Theta_{s}}(\lambda)=\left\lbrace\begin{array}{lll}
(\lambda-2^{\frac{s}{2}})^{\frac{s+2}{2}}(\lambda+2^{\frac{s}{2}})^{\frac{s}{2}}&\textrm{if}& s\quad \textrm{is}\quad \textrm{even}\\
(\lambda-2^{\frac{s}{2}})^{\frac{s+1}{2}}(\lambda+2^{\frac{s}{2}})^{\frac{s+1}{2}}&\textrm{if}& s\quad \textrm{is}\quad \textrm{odd}
\end{array}\right. $
\end{proof}
\color{black}

\subsection{Self-dual NRT-Codes in $M_{n,s}(\mathbb{F}_{2})$ with odd $s$}

The main purpose of this subsection is to use the properties of the shape enumerator of a self-dual NRT-code $C\subseteq M_{n,s}(\mathbb{F}_{2})$ and the properties of the matrix $\Theta_{s}$ to understand how the shape enumerator behaves. 

We recall that Theorem \ref{MacWilliamsH} says that for a self-dual NRT-code $C\subseteq M_{n,s}(\mathbb{F}_{2})$ its $H$-enumerator $H_C$ satisfies the equation
\begin{eqnarray*}
	H_{C}(z_{0},\ldots,z_{s})
	&=&H_{C}\left(\frac{\Theta_{s}}{2^{\frac{s}{2}}}(z_{0},\ldots,z_{s})\right).
\end{eqnarray*} 
In polynomial invariant theory language, this is equivalent to saying that  $H_{C}$ is invariant by 
\begin{equation}\label{Tger}
T=\frac{\Theta_{s}}{2^{\frac{s}{2}}}.
\end{equation}
So $H_C$ will be invariant under the group $G=\langle T\rangle$. Note also that by item a) of Theorem \ref{propT}
\[T^{2}=\left(\frac{\Theta_{s}}{2^{\frac{s}{2}}}\right)^{2}=\frac{\Theta_{s}^{2}}{2^{s}}=I_{s+1},\]
and the group $G$ is given by $G=\langle T\rangle=\lbrace I_{s+1},T\rbrace$. 

Let us calculate the Molien's series of $G$ which tells us what kind of invariants we should look for. We can use the properties c) of Theorem \ref{propT} and Theorem \ref{caracTheta}, since $s$ is odd the Molien's series of $G$ can be written as

\begin{eqnarray*}
	\Phi_{G}(\lambda)
	&=&\frac{1}{2}\sum_{A\in G}\frac{\det(A)}{\det(A-\lambda I_{s+1})}\\
	&=&\frac{1}{2}\left[\frac{(-1)^{\frac{s+1}{2}}}{(-1)^{\frac{s+1}{2}}(1-\lambda)^{\frac{s+1}{2}}(1+\lambda)^{\frac{s+1}{2}}}+\frac{1}{(1-\lambda)^{s+1}}\right]\\
	&=&\frac{(1-\lambda)^{\frac{s+1}{2}}+(1+\lambda)^{\frac{s+1}{2}}}{2(1+\lambda)^{\frac{s+1}{2}}(1-\lambda)^{s+1}}\\
	&=&\frac{\displaystyle\sum_{k=0}^{\frac{s+1}{2}}\binom{\frac{s+1}{2}}{k}\lambda^{k} +\displaystyle\sum_{k=0}^{\frac{s+1}{2}}\binom{\frac{s+1}{2}}{k}(-\lambda)^{k}}
	{2 (1-\lambda)^{\frac{s+1}{2}}(1-\lambda^{2})^{\frac{s+1}{2}}}.
\end{eqnarray*}
Consider the subcase of  $\frac{s+1}{2}$ even, that is,   $\frac{s+1}{2}=2t$ for some $t \geq 0$. 
The Molien's series of $G$ can be rewritten as 
\begin{eqnarray*}
	\Phi_{G}(\lambda)
	&=&\frac{\displaystyle\sum_{k=0}^{2t}\binom{2t}{k}\lambda^{k}+\displaystyle\sum_{k=0}^{2t}\binom{2t}{k}(-\lambda)^{k}}{2(1-\lambda)^{2t}(1-\lambda^{2})^{2t}}\\
	&=&\frac{2\displaystyle\sum_{l=0}^{t}\binom{2t}{2l}\lambda^{2l}}{2(1-\lambda)^{2t}(1-\lambda^{2})^{2t}}.
\end{eqnarray*}
hence in this case
\begin{equation}
\Phi_{G}(\lambda)=\frac{\displaystyle\sum_{l=0}^{t}\binom{2t}{2l}\lambda^{2l}}{(1-\lambda)^{2t}(1-\lambda^{2})^{2t}}.
\end{equation}
Note that the term $(1-\lambda)^{2t}$ in the denominator indicates that to form a good basis we should look for $\frac{s+1}{2}$ invariants of degree one, but by Theorem \ref{invgrau1} and item b) of Theorem \ref{propT} there are no invariant polynomial under $G$ of degree one. So we are not using all the information about the shape enumerator of $C$.

Since the dimension of $C$ is $k=\frac{ns}{2}$ and $s$ is an odd number we must have $n$ even, 
which implies that $H_{C}$ will be invariant by $-I$, so $H_{C}$ is invariant under the action of the group
\[ G_{1}:=\lbrace -I,I,-T,T\rbrace .\]
It is easy to see that the Molien series of $G_{1}$ can be written as
\[\Phi_{G_{1}}(\lambda)=\frac{1}{2}\left(\Phi_{G}(\lambda)+\Phi_{G}(-\lambda)\right)\]
and, so
\begin{eqnarray*}
	\Phi_{G_{1}}(\lambda)
	&=&\frac{1}{2}\left[\frac{\displaystyle\sum_{l=0}^{t}\binom{2t}{2l}\lambda^{2l}}{(1-\lambda)^{2t}(1-\lambda^{2})^{2t}}+ \frac{\displaystyle\sum_{l=0}^{t}\binom{2t}{2l}\lambda^{2l}}{(1+\lambda)^{2t}(1-\lambda^{2})^{2t}}\right] \\
	&=&\frac{1}{2}\frac{\displaystyle\sum_{l=0}^{t}\binom{2t}{2l}\lambda^{2l}\left[2\displaystyle\sum_{l=0}^{t}\binom{2t}{2l}\lambda^{2l}\right]}{(1-\lambda^{2})^{4t}} \\
	&=&\frac{\left[\displaystyle\sum_{l=0}^{t}\binom{2t}{2l}\lambda^{2l}\right]^{2}}{(1-\lambda^{2})^{4t}}.
\end{eqnarray*}

Now, if we consider the subcase where $\frac{s+1}{2}$ is odd, $\frac{s+1}{2}=2t+1$, for some $t \geq 0$, then proceeding in the same way as in the even case one obtains the following expression for the Molien's series: 
\[\Phi_{G_{1}}(\lambda)=\frac{\left[\displaystyle\sum_{l=0}^{t}\binom{2t+1}{2l}\lambda^{2l}\right]^{2}}{(1-\lambda^{2})^{4t+2}}.\]

\noindent In short, we have proved the following result: 
\begin{teo}\label{MOliodd} Let $C\subseteq M_{n,s}(\mathbb{F}_{2})$ be a self-dual NRT-code and suppose that $s$ is an odd number. The shape enumerator of $C$ will be an invariant polynomial under the group $G_{1}:=\left\lbrace I,-I,T,-T\right\rbrace$ where $T$is given by (\ref{Tger}) and the Molien's series of $G_{1}$ is
$$\Phi_{G_{1}}(\lambda)=\left\lbrace\begin{array}{lll}
\frac{\left[\displaystyle\sum_{l=0}^{t}\binom{2t}{2l}\lambda^{2l}\right]^{2}}{(1-\lambda^{2})^{4t}}&\text{if}&\frac{s+1}{2}=2t, t=0,1,\ldots\\
\frac{\left[\displaystyle\sum_{l=0}^{t}\binom{2t+1}{2l}\lambda^{2l}\right]^{2}}{(1-\lambda^{2})^{4t+2}}&\text{if}&\frac{s+1}{2}=2t+1, t=0,1,\ldots
\end{array}\right. .$$
\end{teo}

\noindent Thus, Theorem \ref{MOliodd} gives us an expectation of how many algebraically independent invariants we must find to form a base of invariants for $\mathcal{J}(G_{1})$.

Note that in the case $s=1$ we have $\frac{s+1}{2}=1=2(0)+1$ and so the Molien's series of $G_{1}$ is given by
\begin{equation*}
	\Phi_{G_{1}}(\lambda)=\frac{\left[\displaystyle\sum_{l=0}^{0}\binom{1}{2l}\lambda^{2l}\right]}{(1-\lambda^{2})^{2}}=\frac{1}{(1-\lambda^{2})^{2}},
\end{equation*}
which agrees with \cite{slone}, where it was shown by MacWilliams et al. that the Hamming weight enumerator of a binary self-dual code is a polynomial in two polynomials $p_{1}(x,y)$ and $p_{2}(x,y)$ where $\deg p_{1}=\deg p_{2}=2$. This fact was expected since the Hamming metric coincides with the NRT-metric in the case where $s=1$ and the shape enumerator is the Hamming weight enumerator.

\subsection{Self-dual NRT-Codes in $M_{n,s}(\mathbb{F}_{2})$ with even $s$}

Following the same steps as in the previous subsection we can prove an analogous result for the case of an $s$ even.

\begin{teo}\label{evens}
Let $C\subseteq M_{n,s}(\mathbb{F}_{2})$ be a self-dual NRT-code, such that $s$ is an even number. The shape enumerator of $C$ will be an invariant polynomial under the group $G:=\lbrace I,T\rbrace$ where $T$ is given by (\ref{Tger}) and the Molien's series of $G_{1}$ is

$$\Phi_{G}(\lambda)=\left\lbrace\begin{array}{lll}
	\frac{\displaystyle\sum_{l=0}^{t}\binom{2t}{2l}\lambda^{2l}}{(1-\lambda^{2})^{2t}(1-\lambda)^{2t+1}}&\textrm{if}&\frac{s}{2}=2t, t=1,\ldots\\
	\frac{\displaystyle\sum_{l=0}^{t}\binom{2t+1}{2l}\lambda^{2l}}{(1-\lambda^{2})^{2t+1}(1-\lambda)^{2t+2}}&\textrm{if}&\frac{s}{2}=2t+1, t=0,1,\ldots
	\end{array}\right. .$$
\end{teo}

\noindent Theorem \ref{evens} gives us an expectation of how many algebraically independent invariants we must find to form a base of invariants for $\mathcal{J}(G)$.

Note that for $s=2$ we have $\frac{s}{2}=1=2(0)+1$ and, so the Molien's series of $G$ is given by
\[
\Phi_{G}(\lambda)=\frac{\displaystyle\sum_{l=0}^{0}\binom{1}{2l}\lambda^{2l}}{(1-\lambda^{2})(1-\lambda)^{2}}=\frac{1}{(1-\lambda^{2})(1-\lambda)^{2}},
\]
which matches the result that we found in Section 3.1.
\section{Construction of Self-Dual Codes in the NRT-Metric}
In this last section we present some constructions of self-dual codes in NRT spaces. 
The first one utilizes a self-dual code in Hamming space as its starting point.

\subsection{Self-Dual Codes in NRT Spaces from Self-Dual codes in Hamming Spaces}

\begin{defn} Given a vector $v=(v_{1},\ldots,v_{s-1},v_{s})\in\mathbb{F}_{q}^{s}$ the flip of $v$, denoted by $\flip(v)$, is the vector $\flip(v)=(v_{s},v_{s-1},\ldots,v_{1})\in\mathbb{F}_{q}^{s}$.
\end{defn}

\begin{ob}
Let $\flip:\mathbb{F}_{q}^{s}\longrightarrow\mathbb{F}_{q}^{s}$ be the function taking $v \in\mathbb{F}_{q}^{s}$ to its flip. Then
\begin{itemize}
	\item[a)] For any $s\in\mathbb{N}$, $\flip:\mathbb{F}_{q}^{s}\longrightarrow\mathbb{F}_{q}^{s}$ is a linear operator.
	\item[b)] 
	If $s=1$ then $\flip\equiv I$ where $I$ denotes the identity operator.
	\item[c)] If $\langle,\rangle_H$ is the standard inner product on $\mathbb{F}_{q}^{s}$ and $v,u\in\mathbb{F}_{q}^{s}$, then $\langle \flip(v),\flip(u)\rangle_H=\langle v,u\rangle_H .$ 
\end{itemize}
\end{ob}

In the next theorem we present a construction of a self-orthogonal NRT-code that is derived from a code $C\subseteq\mathbb{F}_{q}^{s}$ over the Hamming space with the standard inner product  $\langle ,\rangle_{H}$.
\begin{teo}\label{const1}
Let $C\subseteq\mathbb{F}_{q}^{s}$ be an $[s,k]$-linear code over the Hamming space and let $C^{\perp}$ its dual code with respect to the standard inner product. Then, the code $C_{o}\subseteq M_{1,2s}(\mathbb{F}_{q})$ given by

\[ C_{o}:=\lbrace (v,\flip(u))\ \text{such that}\ v\in C\ \text{and}\ u\in C^{\perp}\rbrace\]
is a $[2s,k + k^\perp]$-self orthogonal code with respect to the NRT-metric where $k^\perp=\dim(C^{\perp})$.
\end{teo}
\begin{proof}
Indeed, let $\left( v_{1},\flip(u_{1})\right),\left(v_{2},\flip(u_{2})\right)\in C_{o}$ where $v_{1},v_{2}\in C$ and $u_{1},u_{2}\in C^{\perp}$. Then
\begin{eqnarray*}
		\langle\left( v_{1},\flip(u_{1})\right),\left(v_{2},\flip(u_{2})\right)\rangle_{N}
		&=&\langle v_{1},u_{2}\rangle_{H}+\langle \flip(u_{1}),\flip(v_{2})\rangle_{H}\\
		&=&\langle v_{1},u_{2}\rangle_{H}+\langle u_{1},v_{2}\rangle_{H}\\
		&=& 0,
\end{eqnarray*}
which means that $C_{o}\subseteq\left(C_{o}\right)^{\perp}$. It's easy to check that $\dim C_{o}=k+k^\perp$.
\end{proof}

\begin{Exemp} Let $C\subseteq\mathbb{F}_{2}^{2}$ be the $[2,1]$-linear code given by $C:=\lbrace (0,0),(1,0)\rbrace$ than its dual code $C^{\perp}$ is given by $C^{\perp}=\lbrace (0,0),(0,1)\rbrace$ and, the code $C_{o}\subseteq M_{1,4}(\mathbb{F}_{2})$ of Theorem \ref{const1} is the following $[4,2]$-self orthogonal NRT-code 
\[C_{o}:=\lbrace(0,0,0,0),(0,0,1,0),(1,0,0,0),(1,0,1,0)\rbrace.\]
Note that if we consider $C_{o}$ as an $[4,2]$-code over the Hamming space $\mathbb{F}_{2}^{4}$, $C_{o}$ is not a self-orthogonal code since $(1,0,0,0)\in C_{o}$ but $(1,0,0,0)\notin C_{o}$. Note also that $C_{o}$ is not a self-dual NRT-code since $(0,1,0,1)\notin C_{o}$ but $(0,1,0,1)\in C_{o}^{\perp}$.
\end{Exemp}

\begin{teo}\label{Const2a}
Let $C\subseteq\mathbb{F}_{q}^{s}$ be an $[s,k]$-self orthogonal code over the Hamming space. Then, the code 
\[C_{ort}:=\lbrace (v,\flip(v))\in M_{1,2s}(\mathbb{F}_{q})\ \text{such that}\ v\in C\rbrace\]
is an $[2s,k]$-self orthogonal code with respect to the NRT-metric.
\end{teo}
\begin{proof}
Let $\mathbf{v}=(v,\flip(v)),\mathbf{u}=(u,\flip(u))\in C_{ort}$ with $v,u\in C$. Then
\[
\langle\mathbf{v},\mathbf{u}\rangle_{N}
 = \langle (v,\flip(v)),(u,\flip(u))\rangle_{N} 
 = 2\langle v, u\rangle_{H} = 0
\]
since $v,u\in C$ and $C$ is a self-orthogonal code over the Hamming space. Clearly, $\dim C_{ort}=k$, so $C_{ort}$ is an $[2s,k]$-self orthogonal NRT-code. 
\end{proof}

\begin{Exemp}
Let $C$ be the $[3,1]$-self orthogonal code given by $C=\lbrace (0,0,0),(1,1,0)\rbrace$ over the Hamming space $\mathbb{F}_{2}^{3}$ and, so the code $C_{ort}\subseteq M_{1,6}(\mathbb{F}_{2})$ of Theorem \ref{Const2a} is the following $[6,1]$-self orthogonal NRT-code
\[C_{ort}=\lbrace (0,0,0,0,0,0), (1,1,0,0,1,1)\rbrace.\]
\end{Exemp}

\begin{teo}\label{Const2}
Let $C\subseteq\mathbb{F}_{q}^{s}$ be an $[s,k]$-self dual code over the Hamming space. Then, the code $C_{N}\subseteq M_{1,2s}(\mathbb{F}_{q})$ given by
\[ C_{N}:=\lbrace (v,\flip(v^{\prime}))\ \text{such that}\ v,v^{\prime}\in C\rbrace\] 
is an $[2s,2k]$-self dual with respect to the NRT-metric.
\end{teo}
\begin{proof}
Let $\mathbf{v}=(v_{1},\flip(v^{\prime}_{1})),\mathbf{u}=(u_{1},\flip(u^{\prime}_{1}))\in C_{N}$ with $v_{1},v^{\prime}_{1},u_{1},u^{\prime}_{1}\in C$. Then, 
\[
\langle\mathbf{v},\mathbf{u}\rangle_{N}
 = \langle (v_{1},\flip(v^{\prime}_{1})),(u_{1},\flip(u^{\prime}_{1}))\rangle_{N}
 =\langle v_{1}, u_{1}\rangle_{H}+\langle v^{\prime}_{1}, u^{\prime}_{1}\rangle_{H}
 =0
\]	
since $v_{1}, v^{\prime}_{1}, u_{1}, u^{\prime}_{1}\in C$ and $C$ is a self-dual code over the Hamming space. So $C_{N}$ is a self orthogonal NRT-code. $C_{N}$ will be a self-dual NRT-code if $\dim(C_{N})=2k$. Let $\beta :=\lbrace v_{1},\ldots v_{k}\rbrace$ be a basis of $C$. Since $\flip:\mathbb{F}_{q}^{s}\longrightarrow\mathbb{F}_{q}^{s}$ is a linear isomorphism the set $flip(\beta):=\lbrace \flip(v_{1}),\ldots ,\flip(v_{k})\rbrace$ is a basis of $flip(C)$. Define
\[\beta_{N}:=\lbrace (v_{1},0),\ldots ,(v_{k},0),(0,\flip(v_{1})),\ldots ,(0,\flip(v_{k}))\rbrace\]
where $0$ denotes the vector $(0,\ldots ,0)\in\mathbb{F}_{q}^{s}$. This is a basis for $C_N$ which has $2k$ elements, and it follows that  $C_{N}$ is a $[2s,2k]$-NRT self dual code.
\end{proof}

\begin{Exemp}\label{ex1}
Let $C$ be the $[2,1]$-self dual code given by $C:=\lbrace (0,0), (1,1)\rbrace$ over the Hamming space $\mathbb{F}_{2}^{2}$ the code $C_{N}\subseteq M_{1,4}(\mathbb{F}_{2})$ of Theorem $\ref{Const2}$ will be the following $[4,2]$-self dual NRT code
$$C_{N}=\lbrace (0,0,0,0), (0,0,1,1), (1,1,0,0), (1,1,1,1)\rbrace.$$
\end{Exemp}

\begin{Exemp}\label{HammExt}
Consider the $[8,4]$-Extended Hamming code $\hat{\mathcal{H}_{3}}$. It is well knowm that $\hat{\mathcal{H}_{3}}$ is a self-dual code over the Hamming space $\mathbb{F}_{2}^{8}$ and thus the construction of Theorem $\ref{Const2}$ applied to $\hat{\mathcal{H}_{3}}$ will give us an code $C_{N}\subseteq M_{1,16}(\mathbb{F}_{2})$, which is an $[16,8]$- self dual NRT code.
\end{Exemp}

\subsection{Constructions of NRT self-dual codes via generator matrices}
In this subsection we will present some constructions of self-dual NRT-codes
starting from other self-dual NRT-codes. These constructions are inspired by the those introduced by Marka et al. in \cite{VMRS}, where some constructions of self-dual NRT-codes for $n=1$ are given. In order to describe NRT codes by generator matrices we will order lexicographically the entries of an element 
$\mathbf{v}   \in M_{n,s}(\mathbb{F}_{q})$, identifying 
the matrix 
$\mathbf{v} = \vrow \in M_{n,s}(\mathbb{F}_{q})$ with a row vector $(v_1 \mid v_2 \mid \ldots \mid v_n) \in M_{1,ns}(\mathbb{F}_{q}) $.

\begin{defn} A generator matrix for an $[ns,k]$-linear code $C\subseteq M_{n,s}(\mathbb{F}_{q})$ in the $NRT$-space is a matrix $G\in M_{k,ns}(\mathbb{F}_{q})$ whose rows form a basis of $C$. A generator matrix $G\in M_{k,ns}(\mathbb{F}_{q})$ of an $[ns,k]$- linear NRT-code $C\subseteq M_{n,s}(\mathbb{F}_{q})$ can be written as
	\[G=\left[\begin{array}{c|c|c|c|c}
	G_{1}&G_{2}&\cdots&G_{n-1}&G_{n}
	\end{array}\right]\]
where each $G_{i}$ is an $k\times s$ matrix for $i=1,\ldots,n$.
\end{defn}

The main point in the constructions given in \cite{VMRS} is the definition of a flip of a matrix $A\in M_{n,s}(\mathbb{F}_{q})$, which is given by

\begin{defn}
Let $A=(a_{i,j})\in M_{n,s}(\mathbb{F}_{q})$. Then, the flip of $A$ denoted by $\Flip(A)$ is defined by
	\[ \Flip(A)=(a_{ik}),\] 
where $k=s-j+1$ for $1\leqslant i\leqslant n$ and $1\leqslant j\leqslant s$. We denote the transpose of $\Flip(A)$ as $A^{o}$.
\end{defn}

\begin{Exemp}
Let $A\in M_{n,s}(\mathbb{F}_{q})$ given by
\[A=\left[\begin{array}{ccccc}
a_{1,1}&a_{1,2}&\cdots&a_{1,s-1}&a_{1,s}\\
a_{2,1}&a_{2,2}&\cdots&a_{2,s-1}&a_{2,s}\\
\vdots&\vdots&\ddots&\vdots&\vdots\\
a_{n-1,1}&a_{n-1,2}&\cdots&a_{n-1,s-1}&a_{n-1,s}\\
a_{n,1}&a_{n,2}&\cdots&a_{n,s-1}&a_{n,s}
\end{array} \right].\]
Then, $\Flip(A)$ and $A^{o}$ are given respectively by
\[	
Flip(A) = \left[\begin{array}{ccccc}
a_{1,s}&a_{1,s-1}&\cdots&a_{1,2}&a_{1,1}\\
a_{2,s}&a_{2,s-1}&\cdots&a_{2,2}&a_{2,1}\\
\vdots&\vdots&\ddots&\vdots&\vdots\\
a_{n-1,s}&a_{n-1,s-1}&\cdots&a_{n-1,2}&a_{n-1,1}\\
a_{n,s}&a_{n,s-1}&\cdots&a_{n,2}&a_{n,1}
\end{array} \right],\ 
A^{o} =\left[\begin{array}{ccccc}
a_{1,s}&a_{2,s}&\cdots&a_{n-1,s}&a_{n,s}\\
a_{1,s-1}&a_{2,s-1}&\cdots&a_{n-1,s-1}&a_{n,s-1}\\
\vdots&\vdots&\ddots&\vdots&\vdots\\
a_{1,2}&a_{2,2}&\cdots&a_{n-1,2}&a_{n,2}\\
a_{1,1}&a_{2,1}&\cdots&a_{n-1,1}&a_{n,1}
\end{array} \right].
\]	
\end{Exemp}
Note that by the definition of the NRT-metric, in the case $n=1$, a code $C\subseteq M_{1,s}(\mathbb{F}_{q})$ is an $[ns,k]$-self orthogonal NRT-code if and only if $GG^{o}=0$ where $G$ is any generator matrix of the code $C$.

In order to define self-dual NRT-codes by generator matrices, we introduce the ordered flip of a matrix $A\in M_{k,ns}(\mathbb{F}_{q})$.
\begin{defn}
	Let $A=\left[\begin{array}{c|c|c|c|c}
	A_{1}&A_{2}&\cdots&A_{n-1}&A_{n}
	\end{array}\right]$  be an $k\times ns$ matrix. The ordered flip of $A$ is the matrix
	$\Oflip(A):=\left[\begin{array}{c|c|c|c|c}
	\Flip(A_{1})&\Flip(A_{2})&\cdots&\Flip(A_{n-1})&\Flip(A_{n})
	\end{array}\right]$. We denote the transpose of $\Oflip(A)$ by $A^{od}$. Note that
	\[ A^{od}=[\Oflip(A)]^{T}=\left[\begin{array}{c}
	A_{1}^{o}\\
	A_{2}^{o}\\
	\vdots\\
	A_{n-1}^{o}\\
	A_{n}^{o}
	\end{array}\right].\]
\end{defn}

The definition of ordered flip and NRT-metric implies the following remark.
\begin{ob}\label{oflis}
Let $C\subseteq M_{n,s}(\mathbb{F}_{q})$ be an $[ns,k]$-linear NRT-code. Then, $C$ is a self orthogonal NRT-code if and only if $GG^{od}=0$ where $G$ is a generator matrix of $C$.
\end{ob}

\begin{teo}\label{ConstrNRTFl1}
Let $C_{i}$ be an $\left[ ns,k_{i}\right]$-self orthogonal NRT-code for $i=1,2$ and $G^{(i)}=[G^{(i)}_{1}|\ldots|G_{n}^{(i)}]$ be a generator matrix of $C_{i}$. Then, the matrix $G\in M_{k_{1}+k_{2},2ns}(\mathbb{F}_{q})$ defined by
	\[G=\left[\begin{array}{c|c|c|c|c}
	G_{1}^{(1)}&0&\cdots&G_{n}^{(1)}&0\\
	0&G_{1}^{(2)}&\cdots&0&G_{n}^{(2)}
	\end{array}\right]\]
is a generator matrix of an $[2ns, k_{1}+k_{2}]$-self orthogonal NRT-code $C^{N}$.
\end{teo}

\begin{proof}
We will prove that $GG^{od}=0$. Indeed, by definition of ordered flip we have 
\[ OFlip(G)=\left[\begin{array}{c|c|c|c|c}
Flip(G_{1}^{(1)})&0&\cdots&Flip(G_{n}^{(1)})&0\\
0&Flip(G_{1}^{(2)})&\cdots&0&Flip(G_{n}^{(2)})
\end{array}\right],
G^{od}=\left[\begin{array}{cc}
(G_{1}^{(1)})^{o}&0\\
0&(G_{1}^{(2)})^{o}\\
(G_{2}^{(1)})^{o}&0\\
\vdots&\vdots\\
(G_{n}^{(1)})^{o}&0\\
0&(G^{(2)}_{n})^{o}
\end{array}\right]\]
and therefore 
\begin{eqnarray*}
	GG^{od}&=&G_{1}^{(1)}(G_{1}^{(1)})^{o}+\ldots +G_{n}^{(1)}(G_{n}^{(1)})^{o}+[G_{1}^{(2)}(G_{1}^{(2)})^{o}+\ldots +G_{n}^{(2)}(G_{n}^{(2)})^{o}]\\
	&=& G_1G_1^{od} + G_2G_2^{od} = 0,
\end{eqnarray*}
given that $C_{1}$ and $C_{2}$ are self-dual NRT-codes. 
 Now it follows from Remark \ref{oflis} that $C^{N}$ is a self orthogonal NRT-code. It is easy to see that the rows of $G$ form a basis of $C^{N}$, so $\dim C^{N}=k_{1}+k_{2}$ and hence $C^{N}$ is an $[2ns,k_{1}+k_{2}]$-self orthogonal NRT-code.
\end{proof}

\begin{cor}
Let $C_{i}$ be an $\left[ ns, \frac{ns}{2}\right]$-self dual NRT-code for $i=1,2$ and $G^{(i)}=[G^{(i)}_{1}|\ldots|G_{n}^{(i)}]$ be a generator matrix of $C_{i}$. Then, the matrix $G\in M_{ns,2ns}(\mathbb{F}_{q})$ defined by
\[G=\left[\begin{array}{c|c|c|c|c}
G_{1}^{(1)}&0&\cdots&G_{n}^{(1)}&0\\
0&G_{1}^{(2)}&\cdots&0&G_{n}^{(2)}
\end{array}\right]\]
is a generator matrix of an $[2ns, ns]$-self dual NRT-code $C^{N}$.
\end{cor}

\begin{Exemp}\label{exNRT1} In Example $\ref{ex1}$, given the Hamming self-dual code $C=\lbrace (0,0), (1,1)\rbrace$, we obtain the self-dual NRT-code $C_{N}\subseteq M_{1,4}(\mathbb{F}_{2})$ which a generator matrix given by
	\[ G^{(1)}=\left[\begin{array}{cccc}
	1&1&0&0\\
	0&0&1&1
	\end{array}\right].\]
Applying the construction of Theorem \ref{ConstrNRTFl1} we obtain the $[8, 4]$-self-dual NRT-code $C^{N_{1}}\subseteq M_{2,4}(\mathbb{F}_{2})$ defined by the generator matrix 
	\[ G:=\left[\begin{array}{cccc|cccc}
	1&1&0&0&0&0&0&0\\
	0&0&1&1&0&0&0&0\\ 
	0&0&0&0&1&1&0&0\\
	0&0&0&0&0&0&1&1
	\end{array}\right].\]
\end{Exemp}
\begin{Exemp}\label{Exemp32}
Applying again the construction to the linear self dual NRT-code $C^{N_{1}}\subseteq M_{2,4}(\mathbb{F}_{2})$ of Example $\ref{exNRT1}$ we obtain the $[16,8]$-self dual NRT-code $C^{N_{2}}$ given by the generator matrix
	\[ G:=\left[\begin{array}{cccc|cccc|cccc|cccc}
	1&1&0&0&0&0&0&0&0&0&0&0&0&0&0&0\\
	0&0&1&1&0&0&0&0&0&0&0&0&0&0&0&0\\
	0&0&0&0&0&0&0&0&1&1&0&0&0&0&0&0\\
	0&0&0&0&0&0&0&0&0&0&1&1&0&0&0&0\\ 
	0&0&0&0&1&1&0&0&0&0&0&0&0&0&0&0\\
	0&0&0&0&0&0&1&1&0&0&0&0&0&0&0&0\\
	0&0&0&0&0&0&0&0&0&0&0&0&1&1&0&0\\
	0&0&0&0&0&0&0&0&0&0&0&0&0&0&1&1\\
	\end{array}\right].\]
\end{Exemp}

\begin{teo}\label{cons39}
	Let $C_{i}\in M_{n_{i}s_{i}}(\mathbb{F}_{q})$ be an $[n_{i}s_{i},k_{i}]$-self orthogonal NRT-code for $1\leqslant i\leqslant t$ such that $k=k_{1}+\ldots+k_{t}\leqslant\bar{n}\bar{s}$, where $\bar{n}:=\max\lbrace n_{i}\rbrace$ and $\bar{s}:=\max\lbrace s_{i}\rbrace$. Let also $G^{(i)}=[G_{1}^{(i)}|\ldots|G_{n_{i}}^{(i)}]$ be a generator matrix of $C_{i}$. Then the matrix $G\in M_{k,\bar{s}(n_{1}+\ldots+n_{t})}(\mathbb{F}_{q})$ defined by 
	
	\[G=\left[\begin{array}{c|c|c|c|c|c|c|c|c|c|c|c|c}
	\widetilde{G}_{1}^{(1)}&\widetilde{G}_{2}^{(1)}&\ldots&\widetilde{G}_{n_{1}}^{(1)}&0&0&\ldots&0&\ldots&0&0&\ldots&0\\
	0&0&\ldots&0&\widetilde{G}_{1}^{(2)}&\widetilde{G}_{2}^{(2)}&\ldots&\widetilde{G}_{n_{2}}^{(2)}&\ldots&0&0&\ldots&0\\
	\vdots&\vdots&\vdots&\vdots&\vdots&\vdots&\vdots&\vdots&\vdots&\vdots&\vdots&\vdots&\vdots\\
	0&0&\ldots&0&0&0&\ldots&0&\ldots&\widetilde{G}_{1}^{(t)}&\widetilde{G}_{2}^{(t)}&\ldots&\widetilde{G}_{n_{t}}^{(t)}\\
	\end{array}\right]\] 
	is a generator matrix of an $[\bar{s}(n_1+\ldots+n_t),k]$-self orthogonal NRT-code $C_{\star}$, where the matrices $\widetilde{G}_{j_i}^{(i)}\in M_{k_{i}\bar{s}}(\mathbb{F}_{q})$, $1\leqslant i\leqslant t$ and $ 1\leqslant j_i\leqslant n_i$ are given by
	$\widetilde{G}_{j_i}^{(i)}=[G_{j_{i}}^{(i)}]$ if $s_{i}=\bar{s}$ or $\widetilde{G}_{j_i}^{(i)}=[G_{j_{i}}^{(i)} \mid 0]$ if $s_{i}<\bar{s}$, where $0\in M_{k_i \bar{s}-s_i}(\mathbb{F}_{q})$ is the null matrix.
	
\end{teo}

\begin{cor}\label{cons39}
Let $C_{i}$ be an $[ns,k_{i}]$-self orthogonal NRT-code for $1\leqslant i\leqslant t$ such that $k=k_{1}+\ldots+k_{t}\leqslant ns$ and $G^{(i)}=[G_{1}^{(i)}|\ldots|G_{n}^{(i)}]$ be a generator matrix of $C_{i}$. Then the matrix $G\in M_{k,tns}(\mathbb{F}_{q})$ defined by 

\[G=\left[\begin{array}{c|c|c|c|c|c|c|c|c|c|c|c|c}
G_{1}^{(1)}&G_{2}^{(1)}&\ldots&G_{n}^{(1)}&0&0&\ldots&0&\ldots&0&0&\ldots&0\\
0&0&\ldots&0&G_{1}^{(2)}&G_{2}^{(2)}&\ldots&G_{n}^{(2)}&\ldots&0&0&\ldots&0\\
\vdots&\vdots&\vdots&\vdots&\vdots&\vdots&\vdots&\vdots&\vdots&\vdots&\vdots&\vdots&\vdots\\
0&0&\ldots&0&0&0&\ldots&0&\ldots&G_{1}^{(t)}&G_{2}^{(t)}&\ldots&G_{n}^{(t)}\\
\end{array}\right]\] 
is a generator matrix of a $[tns,k]$-self orthogonal NRT-code $C_{\star}$.
\end{cor}

In particular, we can apply the preceding Corollary to construct a self-dual NRT-code $C_{\star}$ equivalent to the code $C^{N}$ obtined by the constructions in Theorem $\ref{ConstrNRTFl1}$ is the folowing
\begin{cor}\label{TeoConstEque}
Let $C_{1}$, $C_{2}$ be two $\left[ns,\frac{ns}{2}\right]$-self dual NRT-codes and let $G^{(1)}=[G_{1}^{(1)}|\ldots|G_{n}^{(1)}]$ be a generator matrix of $C_{1}$, and  $G^{(2)}=[G_{1}^{(2)}|\ldots|G_{n}^{(2)}]$ be a generator matrix of $C_{2}$. Then, the matrix $G\in M_{ns,2ns}(\mathbb{F}_{q})$ defined by
\[G=\left[\begin{array}{c|c|c|c|c|c}
G_{1}^{(1)}&\ldots&G_{n}^{(1)}&0&\ldots&0\\
0&\ldots&0&G_{1}^{(2)}&\ldots&G_{n}^{(2)}
\end{array}\right]\] 
is a generator matrix for an $[2ns, ns]$-self dual NRT-code $C_{\star}$.
\end{cor}

\begin{Exemp} Let $C^{N_{1}}$ be the $[8, 4]$-Self dual NRT-code of Example $\ref{exNRT1}$. Then, by the construction of Theorem $\ref{TeoConstEque}$ the following matrix 
\[ G:=\left[\begin{array}{cccc|cccc|cccc|cccc}
	1&1&0&0&0&0&0&0&0&0&0&0&0&0&0&0\\
	0&0&1&1&0&0&0&0&0&0&0&0&0&0&0&0\\
	0&0&0&0&1&1&0&0&0&0&0&0&0&0&0&0\\
	0&0&0&0&0&0&1&1&0&0&0&0&0&0&0&0\\ 
	0&0&0&0&0&0&0&0&1&1&0&0&0&0&0&0\\
	0&0&0&0&0&0&0&0&0&0&1&1&0&0&0&0\\
	0&0&0&0&0&0&0&0&0&0&0&0&1&1&0&0\\
	0&0&0&0&0&0&0&0&0&0&0&0&0&0&1&1\\
\end{array}\right]\]
is a generator matrix for a $[16,8]$-self dual NRT-code $C_{\star}$ that is equivalent to the code $C^{N_{2}}$ of Example $\ref{Exemp32}$.
\end{Exemp}

\subsection{An application of the ordered flip concept}
In \cite{StandardMarcelo}, Alves gave an analogue for NRT-codes of the well-known standard form of generator matrices for codes in spaces with the Hamming space. In particular, for bidimensional codes it was shown the following result.
\begin{teo}\label{bid}
Let $C\subseteq M_{n,s}(\mathbb{F}_{q})$ be an $[ns,2]$-linear NRT-code. Then, $C$ has a generator matrix of the form
\[G=\left[G_{1}|G_{2}|\ldots|G_{n}\right]\]
where each $G_{i}$ is an $2\times s$ matrix of one of the following types: Null matrix, $\left[\begin{array}{c}
		e_{i}\\
		0
	\end{array}\right]$, $\left[\begin{array}{c}
	0\\
	e_{i}
	\end{array}\right]$, $\left[\begin{array}{c}
	e_{i}\\
	e_{j}
	\end{array}\right]$, $\left[\begin{array}{c}
	e_{i}+\lambda e_{j}\\
	e_{j}
	\end{array}\right]$ where $e_{i}\in\mathbb{F}_{q}^{s}$ denotes the i-th canonical vector, $1\leqslant i\leqslant j$ and $\lambda\neq 0$.
\end{teo} 
The definition of ordered flip and Theorem \ref{bid} can be used to find all possible generator matrices of a bidimensional self-dual NRT-code. In fact, if $C\subseteq M_{n,s}(\mathbb{F}_{q})$ is an $[ns,k]$-self dual bidimensional NRT-code then $2=dim(C)=\frac{ns}{2}$ and so $ns=4$, which implies that one of the following cases holds.
\begin{itemize}
	\item[i)] $n=1$, $s=4$ and an $[4,2]$-Self dual NRT-code $C\subseteq M_{1,4}(\mathbb{F}_{q}$) has as a generator matrix one of the matrices
	\[\left[\begin{array}{cccc}
	1&0&0&0\\
	0&1&0&0
	\end{array}\right],\left[\begin{array}{cccc}
	1&0&0&0\\
	0&0&1&0
	\end{array}\right],\left[\begin{array}{cccc}
	0&1&0&0\\
	0&0&0&1
	\end{array}\right],\left[\begin{array}{cccc}
	0&0&1&0\\
	0&0&0&1
	\end{array}\right],\] 
	\[\left[\begin{array}{cccc}
	1&0&\lambda&0\\
	0&0&1&0
	\end{array}\right],\left[\begin{array}{cccc}
	1&\lambda&0&0\\
	0&1&0&0
	\end{array}\right],\left[\begin{array}{cccc}
	0&1&0&\lambda\\
	0&0&0&1
	\end{array}\right],\left[\begin{array}{cccc}
	0&0&1&\lambda\\
	0&0&0&1
	\end{array}\right].\]
	\item[ii)] $n=2$, $s=2$ and an $[4,2]$-Self dual NRT-code $C\subseteq M_{2,2}(\mathbb{F}_{q})$ has as a generator matrix one of the matrices

\[\left[\begin{array}{cc|cc}
1&0&0&0\\
0&0&1&0
\end{array}\right], \left[\begin{array}{cc|cc}
1&0&0&0\\
0&0&0&1
\end{array}\right], \left[\begin{array}{cc|cc}
0&1&0&0\\
0&0&1&0
\end{array}\right], \left[\begin{array}{cc|cc}
0&1&0&0\\
0&0&0&1
\end{array}\right],\]
\[\left[\begin{array}{cc|cc}
0&1&1&0\\
0&0&1&0
\end{array}\right], \left[\begin{array}{cc|cc}
0&1&0&1\\
0&0&0&1
\end{array}\right], \left[\begin{array}{cc|cc}
1&0&1&0\\
0&0&1&0
\end{array}\right], \left[\begin{array}{cc|cc}
1&0&0&1\\
0&0&0&1
\end{array}\right],\]
\[\left[\begin{array}{cc|cc}
0&1&1+\lambda&0\\
0&0&1&0
\end{array}\right], \left[\begin{array}{cc|cc}
0&1&0&1+\lambda\\
0&0&0&1
\end{array}\right], \left[\begin{array}{cc|cc}
1&0&1+\lambda&0\\
0&0&1&0
\end{array}\right], \left[\begin{array}{cc|cc}
1&0&0&1+\lambda\\
0&0&0&1
\end{array}\right].\]
\item[iii)] $n=4$, $s=1$ and in this case the NRT-metric and the Hamming metric are equivalent. The classification of self-dual codes in this case is given by Pless in \cite{Vpless}.
\end{itemize}
\end{document}